\pdfoutput=1
\documentclass[12pt]{article}
\usepackage[utf8]{inputenc}
\usepackage[margin=2.5cm]{geometry}
\usepackage{amsmath,amssymb,amsthm}
\usepackage{graphicx}
\usepackage{hyperref}
\usepackage{lineno}
\theoremstyle{plain}
\newtheorem{theorem}{Theorem}
\newtheorem{proposition}[theorem]{Proposition}
\newtheorem{lemma}[theorem]{Lemma}
\newtheorem{corollary}[theorem]{Corollary}
\theoremstyle{definition}
\newtheorem{definition}{Definition}
\theoremstyle{remark}
\newtheorem{remark}{Remark}

\newcommand{\ket}[1]{\lvert #1 \rangle}
\newcommand{\bra}[1]{\langle #1 \rvert}
\newcommand{\braket}[2]{\langle #1 \vert #2 \rangle}
\newcommand{\ketbra}[2]{\lvert #1 \rangle\langle #2 \rvert}
\newcommand{\C}{\mathbb C}
\newcommand{\Hs}{\mathcal H}
\newcommand{\dd}{\mathrm d}
\newcommand{\Dom}{\mathcal D}

\title{Coherent-State Integrals Beyond the Domain of Operator Products}
\author{Ghasem Asadi Cordshooli\\[2mm]
\small Department of Physics, YI. C., Islamic Azad University, Tehran, Iran\\
\small E-mail: \texttt{ghascor@iau.ac.ir}}
\date{}

\newcommand{\doi}[1]{\href{https://doi.org/#1}{doi:#1}}
\newcommand{\arxivlink}[1]{\href{https://arxiv.org/abs/#1}{arXiv:#1}}

\begin{document}
\maketitle

\begin{abstract}
Inserting a resolution of the identity between two operators turns a matrix element of their product into a phase-space integral. For unbounded operators the product need not exist on the vector at hand, while the integral may still converge. We follow this gap through the chain $e^{\mu a^2}e^{\lambda a^{\dagger2}}$ of two-photon exponentials, where it can be traced completely. Five regions of parameter space appear: where the Bochner conditions of the strong formulation hold, where the product exists as a vector, where it exists as a sesquilinear form, where the coherent-state integral converges, and where the Fock-space series converges. The first three are strictly nested inside the intersection of the last two, and these two overlap without either containing the other. A single Hermite-type function agrees with the matrix element where it exists and with the integral and the series wherever they converge, so the integral carries the analytic continuation of the matrix element beyond the domain of the product and beyond the convergence of its Fock series. On the real slice of parameter space, an iterated form of the integral also reaches every point where the series converges. The analysis rests on a weak composition lemma whose hypotheses concern only the bra, the ket and the resolving family. For bounded chains the same identity is unconditional, and worked examples of that kind are given in a companion paper.
\end{abstract}


\section{Introduction}
\label{sec:intro}

Continuous resolutions of the identity are a basic tool of mathematical physics. Their modern form goes back to Schr\"odinger's minimum-uncertainty packets \cite{Schrodinger1926}, to Klauder's continuous-representation theory \cite{Klauder1963,Klauder1963_2} and to the coherent-state description of optical coherence by Glauber \cite{Glauber1963} and Sudarshan \cite{Sudarshan1963}. The group-theoretical generalisation is due to Perelomov \cite{Perelomov1986}, and systematic accounts are given in Refs.~\cite{KlauderSudarshan1968,AliAntoineGazeau2014,CombescureRobert2012}. The analytic side of the canonical case is the Segal--Bargmann representation \cite{Bargmann1961,Folland1989,Zhu2012}, in which integral identities for entire functions take a natural form \cite{Carlen1991}. In quantum optics the resolution of the identity underlies the ordered expansions of Cahill and Glauber \cite{CahillGlauber1969} and the technique of integration within ordered products \cite{FanZaidiKlauder1987,Fan2003,FanYuan2010}.

In many computations the resolution is inserted between the factors of a product $F_N\cdots F_1$. A matrix element of the product then becomes an iterated integral whose kernels are the matrix elements $\bra{\gamma_k}F_k\ket{\gamma_{k-1}}$ of the individual factors. This is the discrete skeleton of the coherent-state path integral \cite{Klauder1963_2,DaubechiesKlauder1985,Klauder2011}, and it is also a practical route from the kernels of simple operators to matrix elements of composite ones. For bounded operators the step needs no comment. For unbounded operators it hides a question. The vector $F_1\psi$ need not lie in the domain of the next factor, so the product may not exist on the vector at hand, while the integral that replaces it can still converge. What the integral then represents, and how far beyond the product it reaches, is the subject of this paper.

Earlier work \cite{AsadiMirzaee2026,AsadiCordshooli2026v1} justified the insertion through Bochner integrals and Hille's theorem \cite{DiestelUhl1977,HillePhillips1957}, under which it holds as an identity between vectors. Here we start from the weaker reading, as an identity between matrix elements. A resolution of the identity with a positive measure defines an isometry into $L^2$, and this alone makes a single insertion between operators $A$ and $B$ valid once the bra lies in the domain of $A^*$, the resolving vectors lie in the domain of $A$, and the ket lies in the domain of $B$ (Lemma~\ref{lem:single}). The value obtained is the product in the sense of sesquilinear forms, which exists even when $AB\psi$ does not. Iterating gives the nested form (Theorem~\ref{thm:nested}), and the Bochner conditions reappear as the setting in which the identity holds for all bras at once (Proposition~\ref{prop:strong}).

Operators with Gaussian kernels have a well-developed theory. Howe's oscillator semigroup collects the bounded operators whose kernels are centred Gaussians, with an explicit composition formula and a positivity criterion \cite{Howe1988}. Hilgert identified it with the semigroup of Bargmann, Brunet and Kramer in the Bargmann--Fock representation and determined which functions of Gaussian type on $\C^n$ lie in the Fock space \cite{Hilgert1989,BrunetKramer1980}, and Folland's monograph summarises both \cite{Folland1989}. The composition law extends algebraically beyond the convergence of the composition integral, through a product of symmetric matrices that is associative wherever it is defined \cite{DerezinskiKarczmarczyk2020}. Unbounded solution operators of quadratic evolution equations on Fock spaces, with sharp boundedness criteria, are analysed in Refs.~\cite{AlemanViola2018,Viola2017}; their distinction between weak and strong solution operators concerns the evolution problem and differs from the weak and strong readings of an inserted resolution used here. For bounded Gaussian operators the levels studied below coincide. The present paper follows them through a chain of two unbounded factors with Gaussian kernels, where they separate.

The central example is the chain $e^{\mu a^2}e^{\lambda a^{\dagger2}}$, where $a$ and $a^\dagger$ are the annihilation and creation operators of a single mode. The two factors are the simplest unbounded members of the SU(1,1) family that generates squeezing \cite{Yuen1976,Truax1985}, and the normalisability limits of such exponentials were examined by Fisher, Nieto and Sandberg \cite{FisherNietoSandberg1984}. For this chain every level at which the matrix element $\bra n e^{\mu a^2}e^{\lambda a^{\dagger2}}\ket\alpha$ can exist is accessible in closed form (Proposition~\ref{prop:su11} and figure~\ref{fig:regions}). The Bochner conditions hold on a region $\mathcal R_{\mathrm B}$; the product exists as a vector on a larger region $\mathcal R_{\mathrm P}$; it exists as a form on a still larger region $\mathcal R_{\mathrm W}$; and the coherent-state integral and the Fock-space series converge on two regions $\mathcal R_{\mathrm I}$ and $\mathcal R_{\mathrm F}$ that contain $\mathcal R_{\mathrm W}$ and overlap without either containing the other. A single Hermite-type function ties these levels together (Corollary~\ref{cor:continuation}). On $\mathcal R_{\mathrm I}\setminus\mathcal R_{\mathrm F}$ it is still represented by a convergent integral, while the Fock series of the vacuum amplitude already diverges there. The integral generated by the resolution of the identity thus survives both outside the domain of the operator product and outside the convergence of the Fock expansion. On the real slice, an iterated form of the integral also covers $\mathcal R_{\mathrm F}\setminus\mathcal R_{\mathrm I}$, where the integral diverges absolutely (Proposition~\ref{prop:iterated}).

For bounded chains the same identity holds without conditions, and the content moves entirely to the kernels and the integrals. Worked chains of that kind, including a Kerr phase followed by squeezing or by a cubic phase gate, a two-mode squeezer followed by a beam splitter, and a three-factor heralded circuit that uses two successive insertions, are given in the companion paper \cite{AsadiPhysScr2026}; none of them is repeated here. Single insertions into functions of the creation operator are treated in Ref.~\cite{AsadiPhysScr2026}. The first version \cite{AsadiCordshooli2026v1} of the companion paper combined the strong form of the insertion with worked examples; its current version \cite{AsadiPhysScr2026} develops the examples, and the analytic questions are developed here.

Section~\ref{sec:prelim} sets up resolutions of the identity and their weak and strong readings. Section~\ref{sec:weak} proves the weak composition lemma and its nested form. Section~\ref{sec:su11} follows the chain $e^{\mu a^2}e^{\lambda a^{\dagger2}}$ through its five regions. Section~\ref{sec:numerics} reports numerical checks and Sec.~\ref{sec:discussion} concludes.

\section{Resolutions of the identity}
\label{sec:prelim}

Let $\Hs$ be a separable Hilbert space and $(\Omega,\eta)$ a $\sigma$-finite measure space. The measure $\eta$ is positive throughout. Signed and complex measures also occur in phase-space representations, but the construction below uses positivity at one point, the isometry in Lemma~\ref{lem:iso}, on which everything else rests.

\begin{definition}[Resolution of the identity]
A family $\{\ket\gamma\}_{\gamma\in\Omega}\subset\Hs$ \emph{resolves the identity with respect to $\eta$} if $\gamma\mapsto\braket\gamma\phi$ is measurable for every $\phi\in\Hs$ and
\begin{equation}
\int_\Omega \dd\eta(\gamma)\,|\braket\gamma\phi|^2=\|\phi\|^2\qquad\text{for all }\phi\in\Hs .
\label{eq:roi}
\end{equation}
The pair $(\{\ket\gamma\},\eta)$ is then called a \emph{resolution of the identity}, and $I=\int_\Omega\dd\eta(\gamma)\,\ketbra\gamma\gamma$ is shorthand for \eqref{eq:roi}.
\end{definition}

The measure is part of the data. The coherent states introduced below resolve the identity with respect to $\dd^2\beta/\pi$ and with respect to no other multiple of Lebesgue measure. Families with this property are also known as continuous Parseval frames \cite{AliAntoineGazeau2014}.

The symbol $\int\dd\eta(\gamma)\,\ketbra\gamma\gamma$ admits two readings, and the distinction runs through the paper. In the \emph{weak} reading, an identity containing it is an identity between numbers, asserted for a given bra and ket. In the \emph{strong} reading it is an identity between vectors, with the integral converging in norm as a Bochner integral \cite{DiestelUhl1977}. Definition~1 is weak, and the next lemma records what it gives.

\begin{lemma}[Weak resolution]
\label{lem:iso}
If $(\{\ket\gamma\},\eta)$ is a resolution of the identity, then for all $\chi,\phi\in\Hs$
\begin{equation}
\braket\chi\phi=\int_\Omega\dd\eta(\gamma)\,\braket\chi\gamma\braket\gamma\phi ,
\label{eq:weak}
\end{equation}
and the integral converges absolutely, with $\int|\braket\chi\gamma\braket\gamma\phi|\,\dd\eta\le\|\chi\|\,\|\phi\|$.
\end{lemma}

\begin{proof}
By \eqref{eq:roi} the map $V\phi=\braket{\cdot}{\phi}$ is an isometry $\Hs\to L^2(\Omega,\eta)$. Polarisation gives $\braket\chi\phi=\langle V\chi,V\phi\rangle_{L^2}$, which is \eqref{eq:weak}, and the Cauchy--Schwarz inequality in $L^2$ gives the bound.
\end{proof}

For a densely defined operator $A$ we denote by $A^*$ its adjoint and by $\Dom(A)$ its domain \cite{ReedSimon1980,Kato1995}. We write $\bra\chi A\ket\gamma:=\braket{\chi}{A\gamma}$ for $\gamma\in\Dom(A)$ and $\bra\gamma B\ket\psi:=\braket{\gamma}{B\psi}$ for $\psi\in\Dom(B)$.

The canonical example is the family of Glauber coherent states on $\Omega=\C$ with $\dd\eta=\dd^2\beta/\pi$,
\begin{equation}
\ket\beta=e^{-|\beta|^2/2}\sum_{n\ge0}\frac{\beta^n}{\sqrt{n!}}\ket n,\qquad
\braket\beta\alpha=e^{-\frac12|\alpha|^2-\frac12|\beta|^2+\beta^*\alpha}.
\label{eq:cs}
\end{equation}
For this family, \eqref{eq:roi} is the isometry of the Bargmann transform $\phi\mapsto\tilde\phi(z)=\sum_n\braket n\phi\,z^n/\sqrt{n!}$ onto the Bargmann space $\mathcal F^2$ of entire functions square-integrable against $e^{-|z|^2}\dd^2z/\pi$ \cite{Bargmann1961,Zhu2012}. Under this transform the annihilation operator $a$ acts as $\dd/\dd z$ and the creation operator $a^\dagger$ as multiplication by $z$, and
\begin{equation}
\braket\beta\phi=e^{-|\beta|^2/2}\,\tilde\phi(\beta^*).
\label{eq:bargmann}
\end{equation}

\section{Composition in the weak sense}
\label{sec:weak}

For bounded operators $A$ and $B$, Lemma~\ref{lem:iso} applied to $A^*\chi$ and $B\psi$ gives at once
\[
\bra\chi AB\ket\psi=\int_\Omega\dd\eta(\gamma)\,\bra\chi A\ket\gamma\,\bra\gamma B\ket\psi .
\]
For unbounded operators the left side may fail to exist, since $B\psi$ need not lie in $\Dom(A)$, while the right side can remain an absolutely convergent integral. The next lemma identifies the right side under hypotheses on the bra, the ket and the family alone.

\begin{lemma}[One insertion between two factors]
\label{lem:single}
Let $(\{\ket\gamma\},\eta)$ be a resolution of the identity, let $A$ be densely defined and $B$ any linear operator in $\Hs$. Suppose that
\begin{enumerate}
\item[(i)] $\psi\in\Dom(B)$,
\item[(ii)] $\chi\in\Dom(A^*)$,
\item[(iii)] $\ket\gamma\in\Dom(A)$ for $\eta$-almost every $\gamma$.
\end{enumerate}
Then
\begin{equation}
\braket{A^*\chi}{B\psi}=\int_\Omega\dd\eta(\gamma)\,\bra\chi A\ket\gamma\,\bra\gamma B\ket\psi ,
\label{eq:single}
\end{equation}
with absolute convergence. If in addition $B\psi\in\Dom(A)$, the left-hand side equals $\bra\chi AB\ket\psi$.
\end{lemma}

\begin{proof}
Apply Lemma~\ref{lem:iso} to the pair $A^*\chi$, $B\psi$ and use $\braket{A^*\chi}{\gamma}=\braket{\chi}{A\gamma}$ for $\gamma\in\Dom(A)$. In particular, $\gamma\mapsto\bra\chi A\ket\gamma$ coincides $\eta$-almost everywhere with the measurable function $\gamma\mapsto\braket{A^*\chi}{\gamma}$, so no separate measurability assumption is needed. The last statement is the definition of the adjoint.
\end{proof}

The left-hand side of \eqref{eq:single} is the product $AB$ evaluated as a sesquilinear form \cite{Kato1995}. It exists under (i)--(ii) whether or not $B\psi\in\Dom(A)$, and Sec.~\ref{sec:su11} shows both cases occurring within one family of operators. When $A$ is bounded and everywhere defined, (ii) and (iii) hold automatically and only (i) remains; when both factors are, \eqref{eq:single} holds for all $\chi$ and $\psi$.

To iterate, let $F_1,\dots,F_N$ be linear operators and write $P_k:=F_NF_{N-1}\cdots F_k$ for $2\le k\le N$, with the natural domain of the product.

\begin{theorem}[Weak nested identity]
\label{thm:nested}
Let $(\{\ket\gamma\},\eta)$ be a resolution of the identity and suppose that
\begin{enumerate}
\item[(N1)] $\psi\in\Dom(F_1)$;
\item[(N2)] for each $k=2,\dots,N$, the operator $P_k$ is densely defined and $\chi\in\Dom(P_k^*)$;
\item[(N3)] for each $k=2,\dots,N$, $\ket\gamma\in\Dom(P_k)$ for $\eta$-almost every $\gamma$.
\end{enumerate}
Then
\begin{multline}
\braket{P_2^*\chi}{F_1\psi}
=\int\!\dd\eta(\gamma_1)\Biggl[\int\!\dd\eta(\gamma_2)\cdots\biggl[\int\!\dd\eta(\gamma_{N-1})\,
\bra\chi F_N\ket{\gamma_{N-1}}\bra{\gamma_{N-1}}F_{N-1}\ket{\gamma_{N-2}}\biggr]\cdots\\
\times\bra{\gamma_2}F_2\ket{\gamma_1}\Biggr]\bra{\gamma_1}F_1\ket{\psi},
\label{eq:nested}
\end{multline}
where the outermost integral converges absolutely and, for $\eta$-almost every value of the outer variables, so does each inner one. If moreover $\psi\in\Dom(F_N\cdots F_1)$, the left-hand side equals $\bra\chi F_N\cdots F_1\ket\psi$.
\end{theorem}

\begin{proof}
By Lemma~\ref{lem:single} with $A=P_2$ and $B=F_1$,
\[
\braket{P_2^*\chi}{F_1\psi}=\int\dd\eta(\gamma_1)\,\braket{\chi}{P_2\gamma_1}\,\bra{\gamma_1}F_1\ket\psi .
\]
For $N=2$ this is \eqref{eq:nested}. For $N>2$ and $\gamma_1\in\Dom(P_2)$ we have $F_2\gamma_1\in\Dom(P_3)$ and
\[
\braket{\chi}{P_2\gamma_1}=\braket{\chi}{P_3F_2\gamma_1}=\braket{P_3^*\chi}{F_2\gamma_1}.
\]
Lemma~\ref{lem:single} with $A=P_3$, $B=F_2$ and the vector $\gamma_1$ expands this as an absolutely convergent integral over $\gamma_2$ with kernel $\braket{\chi}{P_3\gamma_2}\bra{\gamma_2}F_2\ket{\gamma_1}$. Repeating the step for $k=3,\dots,N-1$ and using $P_N=F_N$ in the last step gives \eqref{eq:nested}. As in the proof of Lemma~\ref{lem:single}, every kernel $\gamma\mapsto\braket{\chi}{P_k\gamma}$ is measurable without further assumptions. The final statement is again the definition of the adjoint.
\end{proof}

\begin{remark}
Equation~\eqref{eq:nested} is an iterated integral. By Fubini's theorem it becomes a Lebesgue integral over $\Omega^{N-1}$ when the product of the absolute values of the kernels is jointly integrable, a condition that \eqref{eq:nested} itself does not require. The chain of Sec.~\ref{sec:su11} uses a single insertion ($N=2$); a three-factor chain with two insertions is worked out in Ref.~\cite{AsadiPhysScr2026}.
\end{remark}

The strong reading of the insertion, used in Refs.~\cite{AsadiMirzaee2026,AsadiCordshooli2026v1}, rests on Bochner integration and Hille's theorem, and gives a vector identity.

\begin{proposition}[Strong form]
\label{prop:strong}
Let $A$ be closed and densely defined, $\phi\in\Hs$, and $\ket\gamma\in\Dom(A)$ for $\eta$-almost every $\gamma$. Suppose that
\begin{equation}
\int_\Omega\dd\eta(\gamma)\,\|\gamma\|\,|\braket\gamma\phi|<\infty
\qquad\text{and}\qquad
\int_\Omega\dd\eta(\gamma)\,\|A\gamma\|\,|\braket\gamma\phi|<\infty .
\label{eq:bochner}
\end{equation}
Then $\phi\in\Dom(A)$, $A\phi=\int\dd\eta(\gamma)\,A\ket\gamma\braket\gamma\phi$ as a Bochner integral, and
\begin{equation}
\braket{\chi}{A\phi}=\int_\Omega\dd\eta(\gamma)\,\bra\chi A\ket\gamma\braket\gamma\phi\qquad\text{for every }\chi\in\Hs .
\label{eq:strong}
\end{equation}
\end{proposition}

\begin{proof}
The map $\gamma\mapsto\ket\gamma\braket\gamma\phi$ is weakly measurable by Definition~1. So is $\gamma\mapsto A\ket\gamma\braket\gamma\phi$: for $f\in\Hs$ choose $f_k\in\Dom(A^*)$ with $f_k\to f$; then $\braket{f}{A\gamma}=\lim_k\braket{A^*f_k}{\gamma}$ pointwise, a limit of measurable functions. Since $\Hs$ is separable, weak measurability implies strong measurability \cite{DiestelUhl1977}, and the first condition in \eqref{eq:bochner} makes $\gamma\mapsto\ket\gamma\braket\gamma\phi$ Bochner integrable. Its integral equals $\phi$, because bounded functionals commute with the Bochner integral and Lemma~\ref{lem:iso} identifies the resulting scalars. The second condition and Hille's theorem for closed operators \cite{DiestelUhl1977,HillePhillips1957} give $\phi\in\Dom(A)$ and the formula for $A\phi$. Applying $\bra\chi$ gives \eqref{eq:strong}.
\end{proof}

For normalised families such as coherent states, $\|\gamma\|=1$ and the first condition in \eqref{eq:bochner} reduces to $\int|\braket\gamma\phi|\,\dd\eta<\infty$. Proposition~\ref{prop:strong} delivers the insertion for all bras at once and places $\phi$ in $\Dom(A)$. Its conditions weigh the growth of $\|A\gamma\|$ against the coefficients $\braket\gamma\phi$, so they are formulated in terms of the chosen family, whereas membership in $\Dom(A)$ is not. Section~\ref{sec:su11} follows the strong reading, the operator product, the weak reading and the bare integral through one family of unbounded chains, where all four part ways.

\begin{remark}[Converse]
\label{rem:converse}
For a strongly measurable map, Bochner integrability is equivalent to the integrability of its norm \cite{DiestelUhl1977}. Both maps in the proof of Proposition~\ref{prop:strong} are strongly measurable, so the two conditions in \eqref{eq:bochner} are also necessary for the corresponding Bochner integrals to exist. They therefore characterise the strong reading for a given resolving family, and Proposition~\ref{prop:su11}(d),(e) shows that the product $A\phi$ can exist where they fail.
\end{remark}

\section{Beyond the domain of the product}
\label{sec:su11}

\subsection{Operators}

For $\lambda\in\C$ let $B_\lambda:=e^{\lambda a^{\dagger2}}$ be the operator of multiplication by $e^{\lambda z^2}$ on its maximal domain $\Dom(B_\lambda)=\{\tilde g\in\mathcal F^2:\,e^{\lambda z^2}\tilde g\in\mathcal F^2\}$ in the Bargmann representation. With its maximal domain, a multiplication operator on a reproducing-kernel Hilbert space is closed, because point evaluations are continuous. For $|\lambda|<1/2$ the domain contains all polynomials, since $z^ne^{\lambda z^2}\in\mathcal F^2$ exactly when $|\lambda|<1/2$; hence $B_\lambda$ is closed and densely defined. For $|\mu|<1/2$ we define the two-photon annihilation exponential as the adjoint
\begin{equation}
A_\mu:=\bigl(B_{\bar\mu}\bigr)^*,\qquad A_\mu^*=B_{\bar\mu},
\label{eq:Amu}
\end{equation}
which is closed and densely defined. Using \eqref{eq:bargmann} one checks that every coherent state lies in $\Dom(A_\mu)$ with $A_\mu\ket\beta=e^{\mu\beta^2}\ket\beta$. One also checks that $\ket n\in\Dom(A_\mu^*)$ and $\ket\alpha\in\Dom(B_\lambda)$ for all $n$ and $\alpha$ when $|\mu|,|\lambda|<1/2$.

\subsection{Three expressions for one matrix element}

For $\lambda,\mu,\alpha\in\C$ and $n\ge0$ we consider the following three expressions.
\begin{itemize}
\item The \emph{nested integral}
\begin{equation}
J_n:=\frac1\pi\int_\C\dd^2\beta\;\beta^n\exp\bigl(-|\beta|^2+\mu\beta^2+\lambda\beta^{*2}+\alpha\beta^*\bigr),
\label{eq:Jn}
\end{equation}
which is $\sqrt{n!}\,e^{|\alpha|^2/2}$ times the right-hand side of \eqref{eq:single} for $A=A_\mu$, $B=B_\lambda$, $\chi=\ket n$ and $\psi=\ket\alpha$.
\item The \emph{Fock series}
\begin{equation}
\Phi_n:=\sum_{j=0}^\infty\frac{\mu^j}{j!}\,e^{\lambda\partial_\alpha^2}\alpha^{n+2j},
\label{eq:Fn}
\end{equation}
obtained by expanding $A_\mu^*\ket n$ and $B_\lambda\ket\alpha$ in the number basis. Here $e^{\lambda\partial_\alpha^2}\alpha^m$ is a polynomial.
\item The \emph{closed form}
\begin{equation}
C_n:=D^{-1/2}\,e^{\mu\alpha^2/D}\;t^n H_n\!\Bigl(\frac{\alpha}{2Dt}\Bigr),
\qquad D:=1-4\lambda\mu,\qquad t^2=-\frac\lambda D,
\label{eq:Cn}
\end{equation}
with the principal branch of $D^{-1/2}$. The product $t^nH_n(\cdot)$ is a polynomial in $t^2$ and does not depend on the sign of $t$.
\end{itemize}
The closed form $C_n$ is analytic on the domain $\mathcal U=\{(\lambda,\mu):\,1-4\lambda\mu\notin(-\infty,0]\}$.

\begin{proposition}
\label{prop:su11}
Let $\alpha\in\C$ and $n\ge0$.
\begin{enumerate}
\item[(a)] \emph{Integral.} $J_n$ converges absolutely if and only if $|\mu+\bar\lambda|<1$, and then $J_n=C_n$.
\item[(b)] \emph{Fock series.} $\Phi_n$ converges absolutely and locally uniformly on $\{|4\lambda\mu|<1\}$, and there $\Phi_n=C_n$. For $\alpha=n=0$ it diverges when $|4\lambda\mu|>1$.
\item[(c)] \emph{Weak composition.} If $|\lambda|<1/2$ and $|\mu|<1/2$, Lemma~\ref{lem:single} applies to $A_\mu$, $B_\lambda$, $\ket n$, $\ket\alpha$, and
\begin{equation}
\braket{A_\mu^*n}{B_\lambda\alpha}=\frac{e^{-|\alpha|^2/2}}{\sqrt{n!}}\,\Phi_n=\frac{e^{-|\alpha|^2/2}}{\sqrt{n!}}\,J_n .
\label{eq:su11weak}
\end{equation}
\item[(d)] \emph{Strong form.} Under the hypotheses of (c), the conditions \eqref{eq:bochner} for $A=A_\mu$ and $\phi=B_\lambda\ket\alpha$ hold if and only if $|\mu+\bar\lambda|<1/2$.
\item[(e)] \emph{Product.} Under the hypotheses of (c), $B_\lambda\ket\alpha\in\Dom(A_\mu)$ if and only if $|\lambda|<\tfrac12|1-4\lambda\mu|$. In that case $A_\mu B_\lambda\ket\alpha$ has the Bargmann function
\begin{equation}
e^{-|\alpha|^2/2}\,D^{-1/2}\exp\!\Bigl[\frac{\alpha z+\mu\alpha^2+\lambda z^2}{D}\Bigr].
\label{eq:product}
\end{equation}
\end{enumerate}
\end{proposition}

\begin{proof}
(a) Since $\mathrm{Re}(\lambda\beta^{*2})=\mathrm{Re}(\bar\lambda\beta^2)$, the modulus of the integrand is $|\beta|^n\exp\{-|\beta|^2+\mathrm{Re}[(\mu+\bar\lambda)\beta^2]+\mathrm{Re}(\alpha\beta^*)\}$. The maximum of the quadratic part over directions is $(|\mu+\bar\lambda|-1)|\beta|^2$. The integral therefore converges if and only if $|\mu+\bar\lambda|<1$; in the borderline case it diverges along the critical direction for every linear term.

On this convex region $J_n$ is analytic in $(\lambda,\mu)$ by dominated convergence on compact subsets, and $1-4\lambda\mu\notin(-\infty,0]$. Indeed, $4\lambda\mu\ge1$ with $\lambda=\rho e^{i\varphi}$ forces $\mu=\sigma e^{-i\varphi}$ with $\rho\sigma\ge1/4$, whereas $|\mu+\bar\lambda|=\rho+\sigma<1$ implies $\rho\sigma<1/4$. By (b) and (c), proved below, $J_n=\Phi_n=C_n$ on the open set $\{|\lambda|<1/2,\ |\mu|<1/2\}$, which is contained in both $\{|\mu+\bar\lambda|<1\}$ and $\{|4\lambda\mu|<1\}$. The identity theorem on the connected region $|\mu+\bar\lambda|<1$ then gives $J_n=C_n$.

(b) We have
\[
e^{\lambda\partial_\alpha^2}\alpha^m=\sum_{i\le m/2}\frac{m!}{i!\,(m-2i)!}\,\lambda^i\alpha^{m-2i}.
\]
Fix a compact set $K$ of parameters $(\lambda,\mu,\alpha)$ with $|4\lambda\mu|<1$, choose $\epsilon\in(0,1)$ such that $4(|\lambda|+\epsilon)|\mu|\le1-\epsilon$ on $K$, and put $\lambda_\epsilon=|\lambda|+\epsilon$ and $M_j=(n+2j)!\,\lambda_\epsilon^j/j!$. For $m=n+2j$ and $i=j-l$ with $0\le l\le j$, the bounds $j!/(j-l)!\le j^l$ and $(n+2l)!\ge(2l)!$ give
\[
\frac{m!}{i!\,(m-2i)!}\,|\lambda|^i|\alpha|^{m-2i}\le M_j\,|\alpha|^n\,\frac{(j|\alpha|^2/\lambda_\epsilon)^l}{(2l)!},
\]
whose sum over $l$ is at most $M_j|\alpha|^n\cosh(|\alpha|\sqrt{j/\lambda_\epsilon})$. The at most $\lfloor n/2\rfloor$ remaining terms, $i=j+k$ with $k\ge1$, are bounded by $M_j\,\lambda_\epsilon^k|\alpha|^{n-2k}$. Since $(n+2j)!/(j!)^2\le(2j+n)^n\,4^j$, the $j$-th term of \eqref{eq:Fn} is bounded on $K$ by $C_K(2j+n)^n(1-\epsilon)^j\,e^{C_K'\sqrt j}$ with constants $C_K,C_K'$. This gives absolute and locally uniform convergence on $\{|4\lambda\mu|<1\}$, including $\lambda=0$ and $\alpha=0$. For $\alpha=n=0$ the terms are $(2j)!\,(\lambda\mu)^j/(j!)^2\sim(4\lambda\mu)^j/\sqrt{\pi j}$, so the series diverges for $|4\lambda\mu|>1$.

For small parameters, summing \eqref{eq:Fn} against $w^n/n!$ gives $e^{\lambda\partial_\alpha^2}e^{w\alpha+\mu\alpha^2}$, and the action of the heat semigroup on a Gaussian \cite{Widder1975} yields
\begin{equation}
\sum_{n\ge0}\frac{w^n}{n!}\Phi_n=D^{-1/2}\exp\!\Bigl[\frac{w\alpha+\mu\alpha^2+\lambda w^2}{D}\Bigr].
\label{eq:genfun}
\end{equation}
Extracting the coefficient of $w^n$ with $\exp(2xy-y^2)=\sum_nH_n(x)y^n/n!$ and $y=wt$ gives $\Phi_n=C_n$ near the origin. On $\{|4\lambda\mu|<1\}$ we have $\mathrm{Re}\,D>0$, so $C_n$ is analytic there, and this region is connected, being star-shaped with respect to the origin. The identity theorem extends $\Phi_n=C_n$ to all of it.

(c) The hypotheses (i)--(iii) of Lemma~\ref{lem:single} were verified in the previous subsection. The kernels are $\bra nA_\mu\ket\beta=e^{\mu\beta^2}\braket n\beta$ and $\bra\beta B_\lambda\ket\alpha=e^{\lambda\beta^{*2}}\braket\beta\alpha$, and their product, integrated against $\dd^2\beta/\pi$, is $e^{-|\alpha|^2/2}J_n/\sqrt{n!}$.

For the Fock form, the Bargmann function of $A_\mu^*\ket n=B_{\bar\mu}\ket n$ is $e^{\bar\mu z^2}z^n/\sqrt{n!}$. Hence $\braket{m}{A_\mu^*n}=\bar\mu^j\sqrt{m!}/(j!\sqrt{n!})$ for $m=n+2j$, and zero otherwise. The Bargmann function of $B_\lambda\ket\alpha$ is $e^{-|\alpha|^2/2}e^{\lambda z^2+\alpha z}$, whose $m$-th Taylor coefficient is $e^{-|\alpha|^2/2}e^{\lambda\partial_\alpha^2}\alpha^m/m!$. Evaluating the inner product in the number basis gives the first equality in \eqref{eq:su11weak}.

(d) By \eqref{eq:bargmann}, $|\braket\beta\phi|=e^{-|\alpha|^2/2}|e^{\lambda\beta^{*2}+\alpha\beta^*}|e^{-|\beta|^2/2}$, while $\|\beta\|=1$ and $\|A_\mu\beta\|=|e^{\mu\beta^2}|$. The first integral in \eqref{eq:bochner} converges because $|\lambda|<1/2$. The second has the quadratic exponent $-\frac12|\beta|^2+\mathrm{Re}[(\mu+\bar\lambda)\beta^2]$, which is integrable if and only if $|\mu+\bar\lambda|<1/2$.

(e) Write $\phi_\lambda=B_\lambda\ket\alpha$ and let $\psi_\lambda$ be the function \eqref{eq:product}. If $\phi_\lambda\in\Dom(A_\mu)$, then $\braket{m}{A_\mu\phi_\lambda}=\braket{A_\mu^*m}{\phi_\lambda}=e^{-|\alpha|^2/2}C_m/\sqrt{m!}$ by (c) and (b). By \eqref{eq:genfun} these are the number-basis coefficients of $\psi_\lambda$, so $A_\mu\phi_\lambda=\psi_\lambda$, and $\psi_\lambda$ lies in $\mathcal F^2$ only if $|\lambda/D|<1/2$.

For the converse, fix $\mu$ and let $\Lambda_\mu=\{\lambda:\,|\lambda|<1/2,\ |2\lambda|<|1-4\lambda\mu|\}$. For $\mu\ne0$ the second condition reads $|2\lambda|<2|\mu|\,|2\lambda-1/(2\mu)|$ with $2|\mu|<1$, which describes an open Apollonius disc, so $\Lambda_\mu$ is the intersection of two discs and is connected. On $\Lambda_\mu$ the maps $\lambda\mapsto\phi_\lambda$ and $\lambda\mapsto\psi_\lambda$ are holomorphic with values in $\Hs$, because their inner products with any fixed vector are holomorphic by dominated convergence on compact subsets. Near $\lambda=-\bar\mu$, which lies in $\Lambda_\mu$, we have $|\mu+\bar\lambda|<1/2$, so Proposition~\ref{prop:strong} and (d) give $\phi_\lambda\in\Dom(A_\mu)$, and $A_\mu\phi_\lambda=\psi_\lambda$ by the first part. The graph $G$ of $A_\mu$ is closed. For every $(f,g)\in G^\perp$ the holomorphic function $\lambda\mapsto\braket f{\phi_\lambda}+\braket g{\psi_\lambda}$ vanishes near $-\bar\mu$ and hence on $\Lambda_\mu$. Thus $(\phi_\lambda,\psi_\lambda)\in G$ for all $\lambda\in\Lambda_\mu$.
\end{proof}

\subsection{Five regions}

Proposition~\ref{prop:su11} separates five regions of the $(\lambda,\mu)$ space:
\begin{align*}
&\mathcal R_{\mathrm B}:\ |\lambda|,|\mu|<\tfrac12,\ |\mu+\bar\lambda|<\tfrac12; &
&\mathcal R_{\mathrm P}:\ |\lambda|,|\mu|<\tfrac12,\ |\lambda|<\tfrac12|1-4\lambda\mu|;\\
&\mathcal R_{\mathrm W}:\ |\lambda|,|\mu|<\tfrac12; &
&\mathcal R_{\mathrm I}:\ |\mu+\bar\lambda|<1;\qquad
\mathcal R_{\mathrm F}:\ |4\lambda\mu|<1 .
\end{align*}
On $\mathcal R_{\mathrm B}$ the strong reading holds; on $\mathcal R_{\mathrm P}$ the product $A_\mu B_\lambda\ket\alpha$ exists as a vector; on $\mathcal R_{\mathrm W}$ it exists as a form; on $\mathcal R_{\mathrm I}$ the coherent-state integral converges; on $\mathcal R_{\mathrm F}$ the Fock series converges. Proposition~\ref{prop:strong} together with (d) and (e) gives $\mathcal R_{\mathrm B}\subset\mathcal R_{\mathrm P}$, and on $\mathcal R_{\mathrm W}$ both $|\mu+\bar\lambda|<1$ and $|4\lambda\mu|<1$. Hence
\[
\mathcal R_{\mathrm B}\subset\mathcal R_{\mathrm P}\subset\mathcal R_{\mathrm W}\subset\mathcal R_{\mathrm I}\cap\mathcal R_{\mathrm F}.
\]
Figure~\ref{fig:regions} shows the real slice. Each inclusion is strict, and the last two regions are not nested:
\begin{itemize}
\item At $(\lambda,\mu)=(0.2,0.4)$ the product exists, since $|\lambda|=0.2<\frac12|1-0.32|=0.34$, and $A_\mu B_\lambda\ket\alpha$ is the Gaussian vector \eqref{eq:product}. The Bochner conditions fail, since $|\mu+\bar\lambda|=0.6$. The strong reading misses a product that exists.
\item At $\lambda=\mu=0.4$ the product no longer exists as a vector, since $|\lambda|=0.4>\frac12|1-0.64|=0.18$, and the matrix element is only the form value \eqref{eq:su11weak}.
\item At $(\lambda,\mu)=(0.8,0.1)$ the intermediate vector $B_\lambda\ket\alpha$ is not normalisable, yet both $J_n$ and $\Phi_n$ converge and agree with $C_n$.
\item At $(\lambda,\mu)=(0.6,-0.6)$ the integral converges and equals $C_n$, while $|4\lambda\mu|=1.44$ places the point outside $\mathcal R_{\mathrm F}$, where the Fock series of the vacuum amplitude diverges.
\item At $(\lambda,\mu)=(0.8,0.25)$ the Fock series converges to $C_n$, while $|\mu+\bar\lambda|=1.05$ and the integral diverges absolutely; its iterated form over $\mathrm{Im}\,\beta$ first converges to $C_n$ (Proposition~\ref{prop:iterated}).
\end{itemize}

\begin{figure}[t]
\centering
\includegraphics[width=\textwidth]{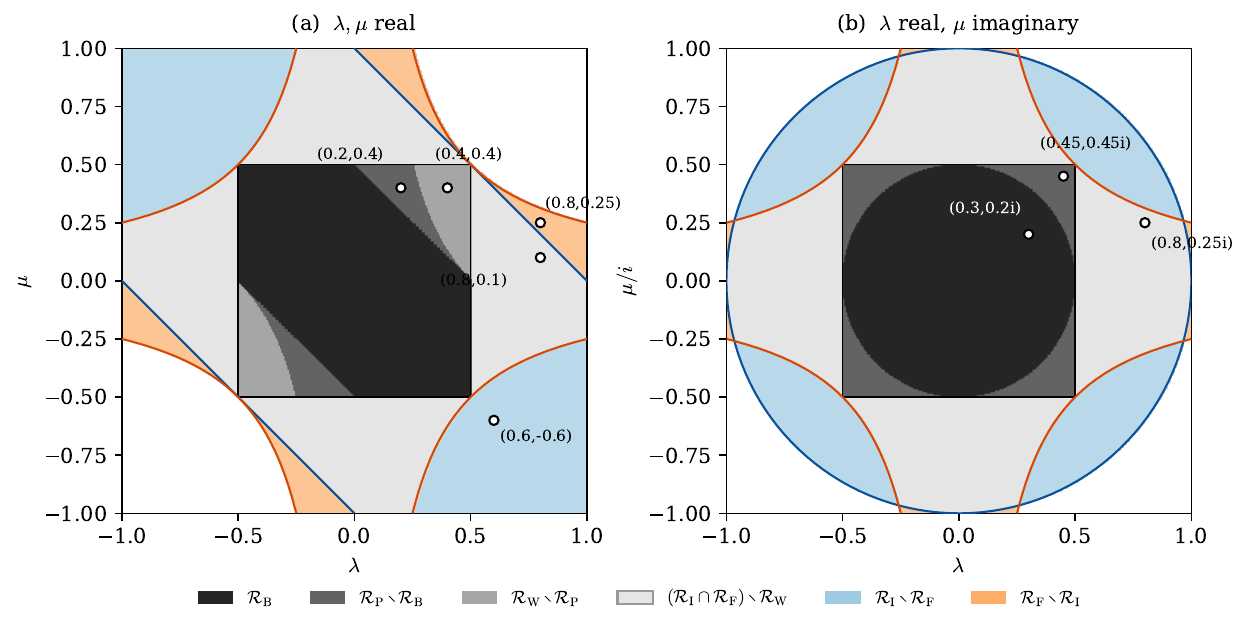}
\caption{Two slices of the regions of Sec.~\ref{sec:su11}. In both panels the square is $\mathcal R_{\mathrm W}$, the blue curve bounds $\mathcal R_{\mathrm I}$ and the hyperbolae $|4\lambda\mu|=1$ bound $\mathcal R_{\mathrm F}$; circles mark points discussed in the text and in Sec.~\ref{sec:numerics}. (a) $\lambda,\mu$ real: $\mathcal R_{\mathrm I}$ is the band $|\lambda+\mu|<1$, and inside the square $\mathcal R_{\mathrm B}$ (black) and $\mathcal R_{\mathrm P}\setminus\mathcal R_{\mathrm B}$ (dark grey) are separated by the line $|\lambda+\mu|=1/2$ and the curve $|\lambda|=\frac12|1-4\lambda\mu|$. The panel is invariant under $(\lambda,\mu)\to(-\lambda,-\mu)$, and only the boundary of $\mathcal R_{\mathrm P}$ breaks the reflection $\lambda\leftrightarrow\mu$ (Remark~\ref{rem:symmetry}). (b) $\lambda$ real and $\mu$ imaginary, with $\mu/i$ on the vertical axis: here $\mathcal R_{\mathrm I}$ and $\mathcal R_{\mathrm B}$ are discs and $\mathcal R_{\mathrm P}$ fills the square.}
\label{fig:regions}
\end{figure}

\begin{remark}[Squeezing and symmetry]
\label{rem:symmetry}
The parameters have a direct meaning. The normally ordered form of the squeeze operator $S(\xi)=\exp[\frac12(\xi^*a^2-\xi a^{\dagger2})]$ with $\xi=re^{i\theta}$ \cite{BarnettRadmore1997} gives $B_\lambda\ket0\propto S(\xi)\ket0$ for $\lambda=-\frac12e^{i\theta}\tanh r$, so $|\lambda|<\frac12$ is finite squeezing and $|\lambda|\to\frac12$ is the limit $r\to\infty$. Since $A_\mu^*=B_{\bar\mu}$, the parameter $\mu$ plays the same role on the side of the bra, and the form value \eqref{eq:su11weak} is the overlap of the squeezed states $B_{\bar\mu}\ket n$ and $B_\lambda\ket\alpha$; for $n=0$ and $\alpha=0$ it equals $(1-4\lambda\mu)^{-1/2}$.

The phase rotation $e^{i\vartheta\hat n}$, where $\hat n=a^\dagger a$, maps $(\lambda,\mu)$ to $(\lambda e^{2i\vartheta},\mu e^{-2i\vartheta})$ and leaves all five regions invariant; on the real slice, $\vartheta=\pi/2$ gives the point symmetry of figure~\ref{fig:regions}. The complex conjugate of $\braket{A_\mu^*\chi}{B_\lambda\psi}$ is $\braket{A_{\bar\lambda}^*\psi}{B_{\bar\mu}\chi}$, so exchanging bra and ket maps $(\lambda,\mu)$ to $(\bar\mu,\bar\lambda)$, which on the real slice is the reflection $\lambda\leftrightarrow\mu$. The conditions defining $\mathcal R_{\mathrm B}$, $\mathcal R_{\mathrm W}$, $\mathcal R_{\mathrm I}$ and $\mathcal R_{\mathrm F}$ are invariant under this exchange. The invariant of the rotation is $\arg(\lambda\mu)$, so panel~(a) of figure~\ref{fig:regions} represents every pair with $\arg(\lambda\mu)\in\{0,\pi\}$, purely imaginary pairs included, while panel~(b) belongs to a different orbit. The condition defining $\mathcal R_{\mathrm P}$ is not, and its boundary is the only curve in figure~\ref{fig:regions} that breaks the reflection.

The reason is visible in \eqref{eq:product}: the output $A_\mu B_\lambda\ket\alpha$ is a squeezed state with parameter $\lambda/D$, and the product exists as a vector exactly when this output squeezing is finite, $|\lambda/D|<\frac12$. Of all the levels, this is the only one that depends on which side carries the ket. At $(\lambda,\mu)=(0.45,0.2)$ the product does not exist as a vector, while at the mirror point $(0.2,0.45)$ it does; both lie in $\mathcal R_{\mathrm W}\setminus\mathcal R_{\mathrm B}$, where the overlap and the integral are available in either order.
\end{remark}

The two panels of figure~\ref{fig:regions} are not related by the phase rotation of Remark~\ref{rem:symmetry}. On the slice of panel~(b), where $\arg(\lambda\mu)=\pi/2$, the region $\mathcal R_{\mathrm I}$ is the disc $|\lambda|^2+|\mu|^2<1$ rather than a band, and $|1-4\lambda\mu|=\sqrt{1+16|\lambda\mu|^2}\ge1$, so $|\lambda|<\frac12\le\frac12|1-4\lambda\mu|$ and the product exists as a vector everywhere in $\mathcal R_{\mathrm W}$: the level $\mathcal R_{\mathrm P}$ becomes invisible. The slice also changes which examples are where. The point $(0.8,0.25i)$ lies in $\mathcal R_{\mathrm I}\cap\mathcal R_{\mathrm F}$, whereas its real counterpart $(0.8,0.25)$ lies outside $\mathcal R_{\mathrm I}$.

The five regions are held together by a single analytic function.

\begin{corollary}
\label{cor:continuation}
Let $\alpha\in\C$, $n\ge0$, and let $M_n(\lambda,\mu)=\sqrt{n!}\,e^{|\alpha|^2/2}\braket{A_\mu^*n}{B_\lambda\alpha}$ on $\mathcal R_{\mathrm W}$. Then $M_n$ is holomorphic on $\mathcal R_{\mathrm W}$, and $C_n$ is its unique holomorphic continuation to the connected open set $\mathcal R_{\mathrm I}\cup\mathcal R_{\mathrm F}$. On $\mathcal R_{\mathrm I}$ the continuation is given by the convergent integral $J_n$, and on $\mathcal R_{\mathrm F}$ by the convergent series $\Phi_n$.
\end{corollary}

\begin{proof}
On $\mathcal R_{\mathrm W}$, $M_n=J_n=\Phi_n=C_n$ by Proposition~\ref{prop:su11}(a)--(c). The region $\mathcal R_{\mathrm I}$ is convex and $\mathcal R_{\mathrm F}$ is star-shaped with respect to the origin, so their union is connected and contains $\mathcal R_{\mathrm W}$. It lies in $\mathcal U$, where $C_n$ is holomorphic, and uniqueness follows from the identity theorem.
\end{proof}

Corollary~\ref{cor:continuation} separates three objects that are easily run together: the matrix element, which exists on $\mathcal R_{\mathrm W}$; its analytic continuation $C_n$; and the two convergent representations of that continuation, the integral on $\mathcal R_{\mathrm I}$ and the series on $\mathcal R_{\mathrm F}$. On $\mathcal R_{\mathrm I}\setminus\mathcal R_{\mathrm F}$ the continuation has a convergent coherent-state integral while the Fock series diverges, as (b) shows for $\alpha=n=0$. The integral generated by the resolution of the identity thus survives outside the domain of the operator product from which it was derived, and outside the convergence of the Fock expansion by which that product is usually computed.

On the real slice the integral reaches further in iterated form.

\begin{proposition}[Iterated integrals on the real slice]
\label{prop:iterated}
Let $\lambda,\mu\in\mathbb R$ with $4\lambda\mu<1$, let $\alpha\in\C$ and $n\ge0$, and write $\beta=x+iy$ in \eqref{eq:Jn}.
\begin{enumerate}
\item[(a)] If $\lambda+\mu>-1$, the iterated integral over $y$ first and then over $x$ converges and equals $C_n$.
\item[(b)] If $\lambda+\mu<1$, the iterated integral over $x$ first and then over $y$ converges and equals $C_n$.
\end{enumerate}
Every real number satisfies $\lambda+\mu>-1$ or $\lambda+\mu<1$, so each point of the real slice of $\mathcal R_{\mathrm F}$ falls under (a) or (b); on $\mathcal R_{\mathrm I}$ both orders agree with the absolutely convergent integral.
\end{proposition}

\begin{proof}
For real $\lambda,\mu$ the quadratic part of the exponent in \eqref{eq:Jn} is
\[
-(1-\lambda-\mu)\,x^2-(1+\lambda+\mu)\,y^2+2i(\mu-\lambda)\,xy .
\]
In case (a) the coefficient of $y^2$ is negative, and the integral over $y$ of a polynomial times this Gaussian converges for every $x$. Completing the square leaves an $x$-integrand equal to a polynomial times
\[
\exp\Bigl[-\frac{1-4\lambda\mu}{1+\lambda+\mu}\,x^2+\ell x\Bigr],
\]
with $\ell$ linear in $\alpha$, which is integrable since $4\lambda\mu<1$. To identify the value, replace $\beta^n$ by $e^{w\beta}$ with $w\in\C$. The two successive Gaussian integrations give
\[
\frac1\pi\int\!dx\int\!dy\;e^{-|\beta|^2+\mu\beta^2+\lambda\beta^{*2}+\alpha\beta^*+w\beta}
=D^{-1/2}\exp\Bigl[\frac{w\alpha+\mu\alpha^2+\lambda w^2}{D}\Bigr],
\]
which is the generating function \eqref{eq:genfun}. The inner integrals converge locally uniformly in $w$, so the $n$-th derivative at $w=0$ passes under the integrals, and the coefficient extraction in the proof of Proposition~\ref{prop:su11}(b) gives $C_n$. Case (b) is the same computation with $x$ and $y$ exchanged, where the second step carries the coefficient $-(1-4\lambda\mu)/(1-\lambda-\mu)$. On $\mathcal R_{\mathrm I}$ Fubini's theorem identifies both orders with the absolutely convergent integral.
\end{proof}

By the phase rotation of Remark~\ref{rem:symmetry}, which rotates the integration variable together with the parameters, the same statement holds for every pair with $\lambda\mu$ real, with the integration axes rotated by the corresponding angle. On these slices the iterated coherent-state integral thus represents the continuation $C_n$ on all of $\mathcal R_{\mathrm F}$, in addition to the absolutely convergent integral on $\mathcal R_{\mathrm I}$.

\section{Numerical verification}
\label{sec:numerics}

Phase-space integrals were computed by the trapezoidal rule on square grids, which converges spectrally for these Gaussian-type integrands. Series and special functions were computed in arbitrary-precision arithmetic, and each deviation below is the largest absolute difference over the quoted parameters and indices. All of them concern the chain $A_\mu B_\lambda$ of Sec.~\ref{sec:su11}, with $\alpha=0.3+0.2i$ and $n=0,\dots,3$:
\begin{itemize}
\item At $(\lambda,\mu)=(0.3,0.2i)$, $(0.3,-0.3+0.4i)$ and $(0.45,0.45)$, which lie in $\mathcal R_{\mathrm W}$, $J_n$ agrees with $C_n$ to between $6\times10^{-14}$ and $1.4\times10^{-13}$, and $\Phi_n$ agrees with $C_n$ to below $10^{-15}$.
\item At $(0.2,0.4)$, which lies in $\mathcal R_{\mathrm P}\setminus\mathcal R_{\mathrm B}$, $J_n$ agrees with $C_n$ to $5\times10^{-14}$ and $\Phi_n$ to $2\times10^{-16}$. The number-basis computation of $A_\mu B_\lambda\ket\alpha$ agrees with the Gaussian vector \eqref{eq:product} to $7\times10^{-18}$ in norm.
\item At $(0.8,0.1)$, where $B_\lambda\ket\alpha\notin\Hs$, $J_n$ agrees with $C_n$ to $1.1\times10^{-13}$ and $\Phi_n$ to $6\times10^{-16}$.
\item At $(0.6,-0.6)$ and $(0.6i,0.55i)$, which lie in $\mathcal R_{\mathrm I}\setminus\mathcal R_{\mathrm F}$, $J_n$ agrees with $C_n$ to $1.1\times10^{-13}$ and $6\times10^{-14}$, respectively.
\item At $(0.8,0.25)$, which lies in $\mathcal R_{\mathrm F}\setminus\mathcal R_{\mathrm I}$, $\Phi_n$ summed to $600$ terms agrees with $C_n$ to $9\times10^{-16}$.
\item At $(0.8,0.25)$ and $(-0.8,-0.25)$, the iterated integrals of Proposition~\ref{prop:iterated}, over $y$ first and over $x$ first respectively, computed by successive Gaussian integration in exact arithmetic, reproduce $C_n$ for $n=0,\dots,3$.
\item At $(-0.5+0.5i,\,-0.4-0.3i)$, outside both regions, the grid approximations to $J_n$ grow without bound as the grid is enlarged.
\end{itemize}

\section{Discussion}
\label{sec:discussion}

A resolution of the identity turns an operator product into an integral of kernels. For bounded factors this is a rewriting. For unbounded factors the rewriting splits into levels that coincide for bounded factors: the strong identity, the product as a vector, the product as a form, the convergent integral and the convergent series. In the chain $e^{\mu a^2}e^{\lambda a^{\dagger2}}$ all of these levels are explicit, and Proposition~\ref{prop:su11} and Corollary~\ref{cor:continuation} show that they are distinct. The weak composition lemma is what gives access to the lower levels. Its hypotheses involve only the bra, the ket and the resolving family, and the value it produces continues analytically into the region where only the integral survives. For bounded chains the same lemma is unconditional, and the content moves to the kernels and the integrals, which is the subject of the companion paper \cite{AsadiPhysScr2026}.

\paragraph{Limitations.} The non-unitary analysis covers the single-mode chain $A_\mu B_\lambda$; longer SU(1,1) chains and other orderings of the factors are not treated. The chain treated here uses a single insertion ($N=2$), and a three-factor chain is worked out in Ref.~\cite{AsadiPhysScr2026}. For $N\ge3$ the joint (Fubini) integrability of the resulting multi-fold integrals, needed to regard them as integrals over $\Omega^{N-1}$, remains to be examined.

\paragraph{Outlook.} Three extensions appear natural. The first is the continuum limit of \eqref{eq:nested} for time-sliced evolutions, where the weak reading may give a simpler route to the coherent-state path integral than the strong one \cite{DaubechiesKlauder1985,Klauder2011}. The second concerns the dependence of the strong conditions on the resolving family. Membership in $\Dom(A_\mu)$ does not depend on the family, while \eqref{eq:bochner} does, and one may ask whether squeezed coherent states, which resolve the identity with the same measure, recover the strong reading on all of $\mathcal R_{\mathrm P}$. The third concerns non-Gaussian unitaries whose coherent-state symbols have borderline growth in the Bargmann space, with the cubic phase gate \cite{AsadiCordshooliPRA} as a natural test case: being unitary, it satisfies the weak conditions trivially, while the strong conditions depend on the growth of its Bargmann symbol.

\section*{Data availability statement}
All data that support the findings of this study are included within the article.


\end{document}